\documentclass[a4paper,UKenglish,cleveref, autoref, thm-restate]{lipics-v2021}

\title{NP-Hardness and a Fixed-Parameter Algorithm for Translocation Distance} 

\titlerunning{NP-Hardness and a Fixed-Parameter Algorithm for Translocation Distance} 

\author{Maria Constantin}{Faculty of Mathematics and Computer Science, University of Bucharest, Romania}{maria.petruta.constantin@drd.unibuc.ro}{}{}
\author{Adrian Micl\u au\c s}{Faculty of Mathematics and Computer Science, University of Bucharest, Romania}{adrian.miclaus1@gmail.com}{}{}
\author{Alexandru Popa}{Faculty of Mathematics and Computer Science, University of Bucharest, Romania}{alexandru.popa@fmi.unibuc.ro}{}{}

\authorrunning{J. Open Access and J.\,R. Public} 

\Copyright{Jane Open Access and Joan R. Public} 



\category{} 

\relatedversion{} 

\nolinenumbers 

\EventEditors{John Q. Open and Joan R. Access}
\EventNoEds{2}
\EventLongTitle{42nd Conference on Very Important Topics (CVIT 2016)}
\EventShortTitle{CVIT 2016}
\EventAcronym{CVIT}
\EventYear{2016}
\EventDate{December 24--27, 2016}
\EventLocation{Little Whinging, United Kingdom}
\EventLogo{}
\SeriesVolume{42}
\ArticleNo{23}

\usepackage{amsthm}
\usepackage{amsmath}
\usepackage{amssymb}
\usepackage[utf8]{inputenc}
\usepackage{color}
\usepackage{algorithm}
\usepackage{algpseudocode}
\usepackage{tikz}
\usepackage{mathtools}
\usepackage{multirow}
\usepackage{comment}

\newtheorem{problem}{Problem}

\begin{document}

\maketitle

\begin{abstract}

In this paper we study the genome rearrangements done by translocation events. 
Genome rearrangements were used to measure evolutionary distance between organisms since 1936 (Dobzhansky and Sturtevant).
The chromosomes are represented as strings of DNA and the \emph{translocation operation} is defined as the exchange of prefixes between two strings. This operation results in the creation of two new strings (chromosomes) that can then be utilized in subsequent translocations. A translocation is referred to as \emph{contiguous} if the new strings are produced in a single copy, so each of them can be used in only one subsequent operation. When the words produced by a translocation operation are considered to have an infinite number of copies, the translocation is referred to as \emph{non-contiguous}. If the exchanged prefixes are of equal length, the translocation is called \emph{uniform}. Otherwise, the translocation is termed \emph{non-uniform}. The \emph{translocation distance} between two sets of strings, termed the input set and the target set, represents the minimum number of translocations necessary to obtain all the strings in the target set via translocation operations.

We prove that both the non-uniform contiguous and the non-uniform
non-contiguous translocation distance problems are NP-hard over arbitrary
finite alphabets, where the alphabet is part of the input. For the case in which the target set consists of a single string, we give a fixed-parameter tractable algorithm parameterized by the length of the target string.
\end{abstract}



\section{Introduction}

\subsection*{Motivation}

The genome encompasses the entirety of DNA molecules in a cell~\cite{hartl2011essential}. Genomes are organized into chromosomes, which are further divided into genes. Each gene is a specific sequence of DNA, which carries the vital information necessary for constructing other molecules~\cite{fertin}. DNA, the hereditary molecule, replicates itself, typically producing exact copies~\cite{hartl2011essential}. However, during replication, mutations may occur, leading to imprecise replicas and alterations in the genome over generations. This process drives the evolution of new species~\cite{zeira2019genome, oliveira2024}. 

Mutations involve the rearrangement of DNA within a gene and may be categorized according to the type of the observed rearrangement into: duplications, deletions, transpositions, inversions, translocations, fusions, fissions, etc. The number of rearrangements needed to derive the genome of one organism from that of another can indicate the evolutionary distance between organisms~\cite{Dobzhansky1936,Dobzhansky1938}. This number of rearrangements is referred to as the genome distance~\cite{fertin}.

The concept of utilizing genome rearrangements as an indicator of evolutionary distance between organisms was introduced by Dobzhansky and Sturtevant~\cite{Dobzhansky1936,Dobzhansky1938} in 1936. Later, in 1941, Sturtevant and Novitski~\cite{sturtevant1941homologies} studied the problem of determining the minimum number of inversions needed to transform one gene sequence into another, with the aim of explaining species evolution. However, their study is limited to sequences of at most nine genes, as examining sequences of ten or more genes using conventional methods becomes extremely challenging. In 1982, Watterson et al.~\cite{WATTERSON19821} overcome this limitation by introducing the use of combinatorial models to formally define the problem of evolutionary distance~\cite{fertin}. Specifically, they use permutations to represent the relative positions of genes along a chromosome. Their technique reduces the problem of determining the evolutionary distance to identifying the minimal number of inversions required to transform one permutation into another. Sankoff and Kececioglu are the first to obtain significant results~\cite{Kececioglu1993ExactAA,kececioglu1995mice,sankoff1992edit,sankoff1992gene,kececioglu1995exact,kececioglu1994efficient} regarding genome rearrangement distance by employing combinatorics. Their research is built upon the combinatorial model proposed by Watterson~\cite{WATTERSON19821} and introduces multiple exact and approximation algorithms for computing the evolutionary distance by inversions.

\subsection*{Informal Problem Definition}
In this paper, we consider the following problem. The input consists of a set of strings $A$. One translocation operation selects two strings from $A$ and two of their prefixes, then extends the set $A$ with two new strings obtained by interchanging the two prefixes. Given two sets of strings $A$ and $B$, the goal is to determine the minimum number of translocations to obtain a set $A'$ such that $B \subseteq A'$.
If, within each translocation operation, the two exchanged prefixes have
equal length, the translocation is called \emph{uniform}, otherwise, it is
called \emph{non-uniform}. The common prefix length may differ between
different translocation operations. A translocation is referred to as \emph{contiguous} if the new strings are produced in a single copy, so each of them can be used in only one subsequent operation. When the strings produced by a translocation operation are considered to have an infinite number of copies, the translocation is referred to as \emph{non-contiguous}.

\subsection*{Previous and Related Work}

Translocations have a long history in genome rearrangement theory, where genomes are typically modeled as unsigned or signed permutations of genes. In an unsigned permutation, each gene is represented only by its identity and its orientation is ignored, whereas in a signed permutation each gene is assigned a positive or negative sign indicating its orientation. In this classical setting, each gene appears exactly once, i.e. duplicates are excluded, and the goal is to transform one genome into another using a minimum number of rearrangement operations. In the classical permutation-based model, a translocation acts on two chromosomes of the current genome: the chromosomes are cut and exchange their corresponding end segments, producing two new chromosomes that replace the original ones. Thus, the chromosomes produced by one operation become part of the current genome and, if used in a subsequent translocation, are consumed and replaced again. In this sense, the classical model is analogous to our contiguous setting, where intermediate strings are available in a single copy. The models are not identical, however, since in our setting the strings in the initial set $A$ are assumed to have infinitely many available copies, and strings may contain repeated symbols.

The first scientific papers that use combinatorics to study the evolutionary distance via translocations are the ones of Kececioglu and Ravi~\cite{kececioglu1995mice}. They investigate the unsigned case of translocation distance. In their paper, Kececioglu and Ravi present a linear-time exact algorithm for the equal-length prefixes case of unsigned translocation distance. They also propose a $2$-approximation algorithm for the case in which
the two exchanged prefixes may have different lengths, with running time
$O(n^2)$, where $n$ denotes the number of genes. Even though Zhu and Wang~\cite{ZHU2006322} prove the NP-hardness of the unsigned translocation distance problem, it continues to be widely studied~\cite{cui20071,cui20081,jiang20141,SilveiraSLA15,pu20201}. Many papers focus on minimizing the performance ratio of the approximation algorithms for this problem, with the previously best ratio achieved thus far being $1.375$~\cite{pu20201}. The $1.375$ approximation frontier has recently been broken by a randomized FPT $(4/3+\varepsilon)$-approximation algorithm~\cite{sun2025randomized}. Hannenhalli's research~\cite{hannenhalli1996polynomial} focuses on the signed case of the translocation distance problem. He proposes a polynomial-time algorithm that runs in $O(n^3)$ and computes the exact signed translocation distance. However, subsequent studies provide improved algorithms for this problem. Wang et al.~\cite{WangZLM05} present an algorithm for signed translocation distance with a time complexity of $O(n^2)$, while Li~\cite{li2004linear} and Bergeron~\cite{bergeron2006sorting} achieve linear-time solutions. Additionally, there are papers which examine genome distance by translocations in conjunction with other operations, including translocations and inversions~\cite{kececioglu1995mice}; translocations, inversions, fusions, and fissions~\cite{hannenhalli1995transforming}; and translocations, inversions, and block interchanges~\cite{yancopoulos2005efficient}. Closely related models also continue to evolve, for instance, symmetric reversals~\cite{lai2025symmetric} and flanked block-interchanges~\cite{li2024flanked}.

A key distinction between the above line of research and our setting is the no-duplicates assumption: the permutation model inherently assumes a single copy of each gene~\cite{zeira2019genome}. Moreover, the contiguous/non-contiguous distinction considered in this paper is not explicitly defined in the classical permutation model.

Motivated by biological scenarios where duplicates cannot be ignored, Martin-Vide and Mitrana~\cite{Martin-VideM04} initiated the study of translocation distance in a string-based framework, in which a translocation swaps prefixes of two strings and the objective is to generate given target strings. Their results concern the uniform contiguous setting: they give an exact algorithm when the target set has size one and a $2$-approximation algorithm for an arbitrary target set~\cite{Martin-VideM04}. Constantin and Popa~\cite{constantin2019some,constantin2025exact} further study the contiguous model, giving an exact algorithm for the uniform translocation distance when the target set has size two and a $2$-approximation algorithm for the non-uniform translocation distance when the target set has size one. The heuristic approaches in~\cite{constantin2024simulated} also concern the non-uniform contiguous translocation distance. To the best of our knowledge, the non-contiguous string-based model has not been studied algorithmically before the present work.

\subsection{Our results}

Our work continues in the string framework of~\cite{Martin-VideM04} and focuses on variants that are not covered by the permutation literature. In particular, we study the non-uniform and non-contiguous formulations at the level of sets of strings.

\subsubsection*{NP-hardness (Section~\ref{sec:np-hard}}

To prove NP-hardness for the non-uniform translocation distance problems, we give a reduction from $(3,B2)$-SAT~\cite{berman}. We construct a target set containing a gadget substring for each variable and each clause, so that the translocation distance reflects the instance size. For every variable, the input set contains two disjoint subsets encoding the assignments \textsc{True} and \textsc{False}. Each subset can produce the same variable gadget, but only the subset consistent with the chosen assignment simultaneously produces exactly the clause strings satisfied by that assignment. The two subsets are designed to be incompatible, so no optimal solution can mix them when building the variable gadget, allowing us to read off the assignment from the performed translocations and complete the reduction.

\subsubsection*{FPT algorithm (Section~\ref{sec:fpt})}
When the target set has size one, i.e.\ $B=\{z\}$, we give an FPT algorithm for the non-uniform contiguous translocation distance parameterized by $q=|z|$. The algorithm has two main phases. First, we perform a kernelization step that replaces the initial input set $A$ by a bounded set $A'$. Each string from $A'$ is an alternating concatenation of pairwise-disjoint substrings of $z$. Intuitively, this abstraction keeps exactly the information relevant to forming $z$. We show (Lemma~\ref{lemma:kernel}) that this transformation preserves the optimum and, moreover, $A'$ is bounded only as a function of $q$. In the second phase, we solve the instance on the kernel by a bounded-depth exhaustive search (Algorithm~\ref{alg:backtrack}). Using Lemma~\ref{lemma:depth}, we never need to explore more than $q{+}1$ translocations. At each search node, we try all ordered pairs of strings in the current set and all boundary cut positions. Combining the depth bound with the kernel-size bound yields an FPT running time with respect to $q$.


\section{Preliminaries}
\label{sec:preliminaries}

In the current paper, we reuse the majority of the notation and definitions introduced by Martin-Vide and Mitrana~\cite{Martin-VideM04}. However, to ease referencing, we reproduce them here as well.

For motivation, chromosomes can be represented linearly as strings over the DNA alphabet $\{A,C,G,T\}$. Nevertheless, all definitions and results in this paper are stated for strings over an arbitrary finite alphabet $\Sigma$. In particular, in our hardness results the alphabet is part of the input and is not assumed to have constant size.

We start with basic notation related to strings. An alphabet $\Sigma$ is a finite, non-empty set of symbols. The empty string is denoted by $\epsilon$, and $\Sigma^{+}=\Sigma^{*}\setminus\{\epsilon\}$. Given a string $w\in\Sigma^*$, we denote by $|w|$ its length. If $w=xy$, where $x,y\in\Sigma^*$, then $x$ is called a prefix of $w$ and $y$ a suffix of $w$. For a string $w$ and $1\leq i\leq j\leq |w|$, $w[i\dots j]$ denotes the substring of $w$ starting at position $i$ and ending at position $j$. If $i>j$, then $w[i\dots j]=\epsilon$.

\begin{definition}[Translocation]
Given two strings $x,y\in\Sigma^+$ and two integers $i,j\in\mathbb{N}$ such that $0\leq i\leq |x|$ and $0\leq j\leq |y|$, the translocation operation swaps the prefix of $x$ of length $i$ with the prefix of $y$ of length $j$. Formally, $(x,y)\vdash_{i,j}(u,v)$ if and only if $x=x_1x_2$, $y=y_1y_2$, $u=y_1x_2$, $v=x_1y_2$, $|x_1|=i$, and $|y_1|=j$.
\end{definition}

A uniform translocation is a translocation in which the two exchanged prefixes have equal length, i.e.\ $i=j$. In contrast, a non-uniform translocation permits prefixes of arbitrary lengths. The common prefix length in a uniform translocation may differ between different translocation operations. Since the upcoming definitions do not depend explicitly on the lengths $i$ and $j$ of the exchanged prefixes, we simplify the notation by writing $\vdash$ instead of $\vdash_{i,j}$.

The translocation operation can be extended to a finite set of strings $A\subseteq\Sigma^+$ to describe the strings that can be obtained by translocating two strings from $A$: $TO(A)=\bigcup_{x,y\in A}\{u,v\mid (x,y)\vdash(u,v)\}$. In this and all subsequent definitions, we assume that every string $x\in A$ has infinitely many available copies and can therefore be used in any number of translocation operations.

The operation can also be iterated to describe all strings that can be
obtained starting from $A$ by repeatedly applying translocations:
$TO_0(A)=A$, $TO_{k+1}(A)=TO(TO_k(A))$, and
$TO_*(A)=\bigcup_{k\geq 0}TO_k(A)$.

Let $S=(s_1,s_2,\dots,s_n)$ be an ordered sequence of translocations, where $s_i=(x_i,y_i)\vdash(u_i,v_i)$ for every $i\in\{1,\dots,n\}$. We refer to $S$ as a translocation sequence. Whether $S$ is executable from the initial set $A$ depends on the availability rules of the contiguous or non-contiguous model, defined below.

Given a translocation sequence $S=(s_1,s_2,\dots,s_n)$, where $s_i=(x_i,y_i)\vdash(u_i,v_i)$, and a string $w\in TO_*(A)$, we define two functions. The value $P_i(S,w)$ denotes the number of copies of $w$ consumed during the first $i$ translocations of $S$, namely $P_i(S,w)=|\{j\leq i\mid x_j=w\}|+|\{j\leq i\mid y_j=w\}|$. The value $F_i(S,w)$ denotes the number of copies of $w$ initially available or produced during the first $i$ translocations of $S$. Since strings in the initial set $A$ are assumed to have infinitely many available copies, we set $F_i(S,w)=\infty$ if $w\in A$; otherwise, $F_i(S,w)=|\{j\leq i\mid u_j=w\}|+|\{j\leq i\mid v_j=w\}|$. For $i=0$, we have $P_0(S,w)=0$, while $F_0(S,w)=\infty$ if $w\in A$ and $F_0(S,w)=0$ otherwise.

A translocation sequence is called contiguous if, at each step, the input strings have enough available copies and each produced copy can be used at most once. Conversely, in the non-contiguous model, once a string has been produced, it may be reused in subsequent operations without restriction. These concepts are formally defined as follows.

\begin{definition}
A translocation sequence $S=(s_1,s_2,\dots,s_n)$ is contiguous from $A$, abbreviated CTS, if for every $i\in\{1,\dots,n\}$ the following conditions hold. If $x_i\neq y_i$, then $F_{i-1}(S,x_i)>P_{i-1}(S,x_i)$ and $F_{i-1}(S,y_i)>P_{i-1}(S,y_i)$. If $x_i=y_i$, then $F_{i-1}(S,x_i)\geq P_{i-1}(S,x_i)+2$.
\end{definition}

\begin{definition}
A translocation sequence $S=(s_1,s_2,\dots,s_n)$ is non-contiguous from $A$, abbreviated NCTS, if, for every $i\in\{1,\dots,n\}$, $F_{i-1}(S,x_i)>0$ and $F_{i-1}(S,y_i)>0$.
\end{definition}

\begin{definition}
Given a CTS $S=(s_1,s_2,\dots,s_n)$ from $A$ and a finite set of strings $B=\{z_1,z_2,\dots,z_m\}$, we say that $S$ is $B$-producing if $F_n(S,z_j)>P_n(S,z_j)$ for every $j\in\{1,\dots,m\}$.
\end{definition}

Thus, a contiguous translocation sequence $S$ is $B$-producing if, after executing all translocations in $S$, at least one available copy of every string in $B$ remains.

\begin{definition}
Given a NCTS $S=(s_1,s_2,\dots,s_n)$ from $A$ and a finite set of strings $B=\{z_1,z_2,\dots,z_m\}$, we say that $S$ is $B$-producing if $F_n(S,z_j)>0$ for every $j\in\{1,\dots,m\}$.
\end{definition}

Thus, a non-contiguous translocation sequence $S$ is $B$-producing if, for every string $z\in B$, either $z$ belongs to the initial set $A$ or at least one copy of $z$ is produced by a translocation in $S$.

The non-uniform contiguous translocation distance from a finite set of strings $A$ to a finite set of strings $B$ is defined as $TD_{NUC}(A,B)=\min\{|S|\mid S \text{ is a contiguous $B$-producing translocation sequence from } A\}$, where arbitrary prefix lengths are allowed in every translocation. Similarly, the non-uniform non-contiguous translocation distance is defined as $TD_{NUNC}(A,B)=\min\{|S|\mid S \text{ is a non-contiguous $B$-producing translocation sequence from } A\}$. If no such $B$-producing sequence exists, the corresponding distance is defined to be $\infty$. If every translocation is instead restricted to exchange prefixes of equal length, then the corresponding translocation distance is called uniform.

We are now ready to introduce the problems studied in this paper.

\begin{problem}[Non-uniform contiguous translocation distance]
\label{problem1}
Given an alphabet $\Sigma$ and two finite sets of strings over it, $A,B\subseteq\Sigma^+$, compute the non-uniform contiguous translocation distance from $A$ to $B$, $TD_{NUC}(A,B)$.
\end{problem}

\begin{problem}[Non-uniform non-contiguous translocation distance]
\label{problem2}
Given an alphabet $\Sigma$ and two finite sets of strings over it, $A,B\subseteq\Sigma^+$, compute the non-uniform non-contiguous translocation distance from $A$ to $B$, $TD_{NUNC}(A,B)$.
\end{problem}

\begin{observation}
Let $\Sigma_A$ be the alphabet of the strings in $A$ and $\Sigma_B$ the alphabet of the strings in $B$. If $\Sigma_B\not\subseteq\Sigma_A$, then the strings in $B$ cannot be obtained from those in $A$. Hence, in the remainder of the paper, we assume that $\Sigma_B\subseteq\Sigma_A$.
\end{observation}

\section{NP-hardness}
\label{sec:np-hard}

We show that the non-uniform contiguous translocation distance problem is NP-hard. We give a reduction from the $(3,B2)$-SAT problem, which is NP-complete~\cite{berman}. In a $(3,B2)$-SAT instance, every variable occurs exactly four times, twice positively and twice negatively.

Let $\phi$ be an instance of $(3,B2)$-SAT with $N$ variables and $M$ clauses. We construct an instance of the non-uniform contiguous translocation distance problem with target set $B=\{z\}$ such that $\phi$ is satisfiable if and only if $TD_{NUC}(A,B)\leq 6N+2M$.

The target string is $z=\$_0gadget(x_1)gadget(x_2)\cdots gadget(x_N)\$c_1\$c_2\$\cdots\$c_M$, where $gadget(x_i)={\mathrm x}_i^1{\mathrm x}_i^2{\mathrm x}_i^3{\mathrm x}_i^4{\mathrm x}_i^5{\mathrm x}_i^6{\mathrm x}_i^7{\mathrm x}_i^8$ for every variable $x_i$, and $c_j={\mathrm a}_j{\mathrm a}_j$ for every clause $C_j$. The symbols $\$_0$ and $\$$ are fresh symbols. In particular, $\$$ is used as a separator between the clause strings.

For every variable $x_i$, let $C_{i_1}$ and $C_{i_2}$ be the two clauses in which $x_i$ occurs positively, and let $C_{i_3}$ and $C_{i_4}$ be the two clauses in which $x_i$ occurs negatively. We construct the following two sets of strings: $
X_i=\{{\mathrm x}_i^1{\mathrm x}_i^4{\mathrm a}_{i_2},\
{\mathrm a}_{i_1}{\mathrm x}_i^3{\mathrm x}_i^2{\mathrm a}_{i_1},\
{\mathrm a}_{i_2}{\mathrm x}_i^5{\mathrm x}_i^8,\
{\mathrm x}_i^7{\mathrm x}_i^6\},
$ and $
\overline{X_i}=\{{\mathrm x}_i^1{\mathrm x}_i^8,\
{\mathrm x}_i^3{\mathrm x}_i^6{\mathrm a}_{i_4},\
{\mathrm a}_{i_3}{\mathrm x}_i^5{\mathrm x}_i^4{\mathrm a}_{i_3},\
{\mathrm a}_{i_4}{\mathrm x}_i^7{\mathrm x}_i^2\}.
$. The initial set is
$
A=\bigcup_{i=1}^{N}X_i\cup\bigcup_{i=1}^{N}\overline{X_i}\cup\{\$,\$_0\}.
$.

Using the strings in $X_i$, the gadget $gadget(x_i)$ and the two clause strings corresponding to the clauses satisfied by setting $x_i=True$ can be obtained in five translocations. More precisely, the following translocations produce $gadget(x_i)$, ${\mathrm a}_{i_1}{\mathrm a}_{i_1}$, and ${\mathrm a}_{i_2}{\mathrm a}_{i_2}$:

\begin{enumerate}
    \item $({\mathrm x}_i^1{\mathrm x}_i^4{\mathrm a}_{i_2},
    {\mathrm a}_{i_1}{\mathrm x}_i^3{\mathrm x}_i^2{\mathrm a}_{i_1})
    \vdash_{|{\mathrm x}_i^1|,|{\mathrm a}_{i_1}{\mathrm x}_i^3|}
    ({\mathrm a}_{i_1}{\mathrm x}_i^3{\mathrm x}_i^4{\mathrm a}_{i_2},
    {\mathrm x}_i^1{\mathrm x}_i^2{\mathrm a}_{i_1})$;

    \item $({\mathrm a}_{i_2}{\mathrm x}_i^5{\mathrm x}_i^8,
    {\mathrm x}_i^7{\mathrm x}_i^6)
    \vdash_{|{\mathrm a}_{i_2}{\mathrm x}_i^5|,|{\mathrm x}_i^7|}
    ({\mathrm x}_i^7{\mathrm x}_i^8,
    {\mathrm a}_{i_2}{\mathrm x}_i^5{\mathrm x}_i^6)$;

    \item $({\mathrm x}_i^1{\mathrm x}_i^2{\mathrm a}_{i_1},
    {\mathrm a}_{i_1}{\mathrm x}_i^3{\mathrm x}_i^4{\mathrm a}_{i_2})
    \vdash_{|{\mathrm x}_i^1{\mathrm x}_i^2|,|{\mathrm a}_{i_1}|}
    ({\mathrm a}_{i_1}{\mathrm a}_{i_1},
    {\mathrm x}_i^1{\mathrm x}_i^2{\mathrm x}_i^3{\mathrm x}_i^4{\mathrm a}_{i_2})$;

    \item $({\mathrm x}_i^1{\mathrm x}_i^2{\mathrm x}_i^3{\mathrm x}_i^4{\mathrm a}_{i_2},
    {\mathrm a}_{i_2}{\mathrm x}_i^5{\mathrm x}_i^6)
    \vdash_{|{\mathrm x}_i^1{\mathrm x}_i^2{\mathrm x}_i^3{\mathrm x}_i^4|,|{\mathrm a}_{i_2}|}
    ({\mathrm a}_{i_2}{\mathrm a}_{i_2},
    {\mathrm x}_i^1{\mathrm x}_i^2{\mathrm x}_i^3{\mathrm x}_i^4{\mathrm x}_i^5{\mathrm x}_i^6)$;

    \item $({\mathrm x}_i^1{\mathrm x}_i^2{\mathrm x}_i^3{\mathrm x}_i^4{\mathrm x}_i^5{\mathrm x}_i^6,
    {\mathrm x}_i^7{\mathrm x}_i^8)
    \vdash_{|{\mathrm x}_i^1{\mathrm x}_i^2{\mathrm x}_i^3{\mathrm x}_i^4{\mathrm x}_i^5{\mathrm x}_i^6|,0}
    (\epsilon,gadget(x_i))$.
\end{enumerate}

Similarly, using the strings in $\overline{X_i}$, the gadget $gadget(x_i)$ and the two clause strings corresponding to the clauses satisfied by setting $x_i=False$ are obtained in five translocations:

\begin{enumerate}
    \item $({\mathrm x}_i^1{\mathrm x}_i^8,
    {\mathrm a}_{i_4}{\mathrm x}_i^7{\mathrm x}_i^2)
    \vdash_{|{\mathrm x}_i^1|,|{\mathrm a}_{i_4}{\mathrm x}_i^7|}
    ({\mathrm a}_{i_4}{\mathrm x}_i^7{\mathrm x}_i^8,
    {\mathrm x}_i^1{\mathrm x}_i^2)$;

    \item $({\mathrm x}_i^3{\mathrm x}_i^6{\mathrm a}_{i_4},
    {\mathrm a}_{i_3}{\mathrm x}_i^5{\mathrm x}_i^4{\mathrm a}_{i_3})
    \vdash_{|{\mathrm x}_i^3|,|{\mathrm a}_{i_3}{\mathrm x}_i^5|}
    ({\mathrm a}_{i_3}{\mathrm x}_i^5{\mathrm x}_i^6{\mathrm a}_{i_4},
    {\mathrm x}_i^3{\mathrm x}_i^4{\mathrm a}_{i_3})$;

    \item $({\mathrm x}_i^1{\mathrm x}_i^2,
    {\mathrm x}_i^3{\mathrm x}_i^4{\mathrm a}_{i_3})
    \vdash_{|{\mathrm x}_i^1{\mathrm x}_i^2|,0}
    (\epsilon,{\mathrm x}_i^1{\mathrm x}_i^2{\mathrm x}_i^3{\mathrm x}_i^4{\mathrm a}_{i_3})$;

    \item $({\mathrm x}_i^1{\mathrm x}_i^2{\mathrm x}_i^3{\mathrm x}_i^4{\mathrm a}_{i_3},
    {\mathrm a}_{i_3}{\mathrm x}_i^5{\mathrm x}_i^6{\mathrm a}_{i_4})
    \vdash_{|{\mathrm x}_i^1{\mathrm x}_i^2{\mathrm x}_i^3{\mathrm x}_i^4|,|{\mathrm a}_{i_3}|}
    ({\mathrm a}_{i_3}{\mathrm a}_{i_3},
    {\mathrm x}_i^1{\mathrm x}_i^2{\mathrm x}_i^3{\mathrm x}_i^4{\mathrm x}_i^5{\mathrm x}_i^6{\mathrm a}_{i_4})$;

    \item $({\mathrm x}_i^1{\mathrm x}_i^2{\mathrm x}_i^3{\mathrm x}_i^4{\mathrm x}_i^5{\mathrm x}_i^6{\mathrm a}_{i_4},
    {\mathrm a}_{i_4}{\mathrm x}_i^7{\mathrm x}_i^8)
    \vdash_{|{\mathrm x}_i^1{\mathrm x}_i^2{\mathrm x}_i^3{\mathrm x}_i^4{\mathrm x}_i^5{\mathrm x}_i^6|,|{\mathrm a}_{i_4}|}
    ({\mathrm a}_{i_4}{\mathrm a}_{i_4},gadget(x_i))$.
\end{enumerate}

We first establish a global lower bound for any sequence producing the target string. Consider any contiguous translocation sequence producing $z$. Trace every occurrence of a symbol in the final copy of $z$ back to the particular copy of an initial string from which that occurrence originates, and let $c$ denote the number of such contributing initial copies.

For each variable $x_i$, the target contains the eight distinct symbols ${\mathrm x}_i^1,\dots,{\mathrm x}_i^8$. Every initial string contains at most two symbols of $gadget(x_i)$, and no initial string contains gadget symbols belonging to two different variables. Hence at least four contributing initial copies are required for every variable. Moreover, the target contains $M$ occurrences of the separator $\$$, each of which must originate from a copy of the one-character initial string $\$$, and its first symbol $\$_0$ must originate from a copy of $\$_0$. Therefore, $c\geq4N+M+1$.

At any point in the sequence, consider only the symbol occurrences that eventually belong to the final copy of $z$. Let $r$ be the number of currently available strings containing at least one such occurrence, and let $g$ be the number of adjacent pairs of tracked occurrences that are already consecutive in the current strings and are consecutive in the same order in $z$. Define the potential $\Phi=g+(c-r)$.

Initially, $r=c$. Furthermore, by construction, no pair of tracked occurrences that are consecutive in $z$ is already consecutive in an initial contributing string. Hence initially $\Phi=0$. In the final state, all tracked occurrences belong to the single string $z$, so $r=1$, while all $|z|-1=8N+3M$ adjacencies of $z$ are present. Therefore the final value of the potential is $\Phi=(8N+3M)+(c-1)\geq12N+4M$.

A single translocation increases $\Phi$ by at most $2$. Indeed, a translocation creates at most two new adjacencies. If $r$ does not decrease, then $c-r$ does not increase, and therefore the potential increases by at most $2$. If $r$ decreases by one, then only one of the two output strings contains tracked occurrences from both input strings, so at most one new adjacency of $z$ can be created, while $c-r$ increases by one. Again, the total increase is at most $2$. Consequently, every contiguous translocation sequence producing $z$ contains at least $6N+2M$ operations.

\begin{lemma}\label{lem:mixed-gadget}
Fix a variable $x_i$ and consider the eight occurrences
${\mathrm x}_i^1,\dots,{\mathrm x}_i^8$ that form $gadget(x_i)$ in the
final target string. Suppose that these eight occurrences originate from
exactly four copies of strings in $X_i\cup\overline{X_i}$, with each copy
contributing exactly two gadget symbols. If these four copies are not all
from $X_i$ and are not all from $\overline{X_i}$, then at least six
translocations are required before the eight tracked occurrences can occur
in the order
${\mathrm x}_i^1{\mathrm x}_i^2\cdots{\mathrm x}_i^8$
in a single string.
\end{lemma}

\begin{proof}
Trace only the eight occurrences that eventually form $gadget(x_i)$ in the
target. For every current string $w$, let $\pi_i(w)$ be the word obtained
by retaining only these tracked occurrences, replacing
${\mathrm x}_i^j$ by $j$, and deleting every other symbol. In particular,
a string that contains no tracked occurrence has projection $\epsilon$,
even if it contains other copies of symbols from
$\{{\mathrm x}_i^1,\dots,{\mathrm x}_i^8\}$.

A translocation on two current strings induces a translocation on their
projections, possibly with one projection equal to $\epsilon$. Indeed, a
cut in an original string determines a cut between two consecutive
tracked occurrences in its projection, and deleting the untracked symbols
preserves the relative order of all tracked occurrences. Therefore, any
lower bound obtained in the projected process also applies to the original
translocation sequence.

The possible projections of the strings in
$X_i\cup\overline{X_i}$ that contain two gadget symbols are
$14,32,58,76,18,36,54,72$. Since the four contributing copies together
contain each of the symbols $1,\dots,8$ exactly once, their projections
must form a partition of $\{1,\dots,8\}$. There are exactly four such
partitions:
$\{14,32,58,76\}$,
$\{18,36,54,72\}$,
$\{14,58,36,72\}$, and
$\{18,54,32,76\}$.
The first is obtained by choosing the four strings from $X_i$, and the
second by choosing the four strings from $\overline{X_i}$. Hence the two
mixed configurations are
$\mathcal M_1=\{18,54,32,76\}$ and
$\mathcal M_2=\{14,58,36,72\}$.

For a projected state, let $r_i$ be the number of nonempty projected
strings and let $g_i$ be the number of correct adjacencies among
$12,23,34,45,56,67,78$ that occur in the projected strings. Define
$\Phi_i=g_i+(4-r_i)$.
For either mixed initial configuration, $r_i=4$ and $g_i=0$, so
$\Phi_i=0$. When the tracked occurrences form the single string
$12345678$, we have $r_i=1$ and $g_i=7$, and hence $\Phi_i=10$.

Every translocation increases $\Phi_i$ by at most $2$. If the number of
nonempty projected strings does not decrease, the translocation creates at
most two new adjacencies, so the increase is at most $2$. If $r_i$
decreases by one, then the two nonempty input projections are merged into
one nonempty output projection; in this case at most one new correct
adjacency is created, while $4-r_i$ increases by one. Thus the total
increase is again at most $2$. If $r_i$ increases, the term $4-r_i$
decreases, so the increase is also at most $2$.

Consequently, if five translocations were sufficient, then every one of
the five translocations would have to increase $\Phi_i$ by exactly $2$.

Consider first $\mathcal M_1=\{18,54,32,76\}$. We enumerate all choices
of two projected input strings, also allowing one of them to be $\epsilon$,
and all cut positions in the two inputs. No first translocation increases
$\Phi_i$ by $2$; the maximum possible increase is $1$. Hence
$\mathcal M_1$ cannot produce $12345678$ in five translocations.

Now consider $\mathcal M_2=\{14,58,36,72\}$. An exhaustive check of the
possible input pairs and cut positions shows that the only projected
states reachable by a first translocation that increases $\Phi_i$ by
exactly $2$ are
$\{1458,36,72\}$,
$\{14,58,3672\}$, and
$\{14,58,7236\}$.
From these states, the only second translocations that again increase
$\Phi_i$ by exactly $2$ lead to one of the two states
$\{1458,3672\}$ or
$\{1458,7236\}$.
For each of these two states, checking all possible input pairs and cut
positions shows that no translocation increases $\Phi_i$ by $2$.
Therefore a sequence of five translocations cannot increase $\Phi_i$ from
$0$ to $10$.

Thus neither mixed configuration can form the eight tracked gadget
occurrences in five translocations. Hence any mixed selection requires at
least six translocations.
\end{proof}

\begin{theorem}
\label{thm:np_complete}
The non-uniform contiguous translocation distance problem is NP-hard over arbitrary finite alphabets, even for a target set of size one.
\end{theorem}

\begin{proof}
Let $\phi$ be an instance of $(3,B2)$-SAT with $N$ variables and $M$ clauses, and construct $A$ and $B=\{z\}$ as above.

Assume first that $\phi$ is satisfiable, and let $\alpha$ be a satisfying assignment. For each variable $x_i$, if $\alpha(x_i)=True$, we apply the five translocations described above to the strings in $X_i$. If $\alpha(x_i)=False$, we instead apply the five translocations described for $\overline{X_i}$. After $5N$ translocations, all $N$ gadget strings have been obtained. Moreover, since $\alpha$ satisfies every clause, for every clause $C_j$ at least one of the selected variable constructions produces the string $c_j={\mathrm a}_j{\mathrm a}_j$.

Starting with $\$_0$, we use $N$ additional translocations to append the gadget strings in the order $gadget(x_1),\dots,gadget(x_N)$. Finally, for each clause, we use one copy of $\$$ and the already produced string $c_j$ to append the substring $\$c_j$ using two translocations. Thus $z$ is obtained in $5N+N+2M=6N+2M$ translocations, and therefore $TD_{NUC}(A,B)\leq6N+2M$.

Conversely, assume that $TD_{NUC}(A,B)\leq6N+2M$, and let $S$ be a $B$-producing sequence of at most $6N+2M$ translocations. By the potential argument above, every sequence producing $z$ contains at least $6N+2M$ translocations. Hence $|S|=6N+2M$.

All inequalities in the lower-bound argument must therefore be tight. In particular, the number of contributing initial copies is exactly $c=4N+M+1$. Hence, for every variable $x_i$, exactly four initial copies contribute the eight occurrences ${\mathrm x}_i^1,\dots,{\mathrm x}_i^8$, and each of the four copies contributes exactly two gadget symbols. There are no additional contributing copies available for the gadget symbols or for the clause symbols.

By Lemma~\ref{lem:mixed-gadget}, the four contributing copies for each variable cannot mix strings from $X_i$ and $\overline{X_i}$ in a translocation sequence attaining the lower bound. Therefore, for each variable $x_i$, the four contributing copies are either all from $X_i$ or all from $\overline{X_i}$.

We define an assignment $\alpha$ by setting $x_i=True$ when the four contributing copies belong to $X_i$, and $x_i=False$ when they belong to $\overline{X_i}$.

It remains to show that $\alpha$ satisfies every clause. Since $c=4N+M+1$, the only contributing initial copies other than the four copies selected for every variable are the $M$ copies of $\$$ and the single copy of $\$_0$. Therefore, both occurrences of every symbol ${\mathrm a}_j$ in the target substring $c_j={\mathrm a}_j{\mathrm a}_j$ must originate from the variable copies already selected above.

If the selected copies for $x_i$ come from $X_i$, then the clause symbols contributed by these copies in pairs correspond to $C_{i_1}$ and $C_{i_2}$, precisely the clauses satisfied by setting $x_i=True$. If the selected copies come from $\overline{X_i}$, the corresponding pairs are those associated with $C_{i_3}$ and $C_{i_4}$, precisely the clauses satisfied by setting $x_i=False$.

Since the target contains $c_j={\mathrm a}_j{\mathrm a}_j$ for every clause $C_j$, at least one selected variable construction must supply the two occurrences of ${\mathrm a}_j$. By the definition of $\alpha$, that variable satisfies $C_j$. Hence every clause is satisfied and $\alpha$ satisfies $\phi$.

Therefore $\phi$ is satisfiable if and only if $TD_{NUC}(A,B)\leq6N+2M$. The reduction is polynomial, and the theorem follows.
\end{proof}


\begin{theorem}
\label{thm:np_complete_non_contiguous}
The non-uniform non-contiguous translocation distance problem is NP-hard
over arbitrary finite alphabets, even for a target set of size one.
\end{theorem}
\begin{proof}
We use the same reduction from $(3,B2)$-SAT as in the contiguous case. Every contiguous translocation sequence is also a non-contiguous one. For the converse direction, observe that the lower bounds used above depend only on the translocations required to create the necessary adjacencies in the target, and not on whether produced strings can be reused. Hence, each variable gadget still requires at least $5$ translocations, while mixing strings from $X_i$ and $\overline{X_i}$ requires at least $6$.

Therefore, any non-contiguous sequence of at most $6\cdot N+2\cdot M$ translocations must have the same structure as in the contiguous case. The same argument then yields a satisfying assignment for the $(3,B2)$-SAT instance. Thus, the non-uniform non-contiguous translocation distance problem is NP-hard.
\end{proof}

\section{FPT algorithm}\label{sec:fpt}

In this section we give an FPT algorithm for the non-uniform contiguous translocation distance with $|B|=1$, i.e.\ $B=\{z\}$, by kernelization. We construct a new input set $A'$ such that $|A'|\le 2\cdot e^{4\cdot(|z|+1)^2}$, $\max{|x|\mid x\in A'}\le 2\cdot |z|+1$, and $TD_{NUC}(A,\{z\})=TD_{NUC}(A',\{z\})$, where $e$ is the base of the natural logarithm. Let $\mathit{SB}=\{z[i\dots j]\mid 1\le i\le j\le |z|\}$ be the set of all substrings of $z$. The kernel keeps only the information relevant for producing $z$: which pairwise-disjoint substrings from $\mathit{SB}$ occur in each string and in what left-to-right order. All other symbols are replaced by a separator $\#$ (with consecutive separators merged). Each valid choice of pairwise-disjoint substrings occurring in an input string may yield a canonical string (Definition~\ref{def:canonical}), so a single input string may witness several canonical strings, and $A'$ contains exactly the canonical strings witnessed by some string in $A$ (Algorithm~\ref{alg:A'}). To prove correctness, fix an optimal translocation sequence $S$ producing $z$ from $A$ and consider the minimal sets of still-needed strings before each step. We transform $S$ into a sequence $S'$ over canonical strings, without increasing its length, by reverse induction: we first replace the last step by one producing $z$ from canonical inputs, and then replace steps backward so that each new step outputs canonical strings matching the canonical forms required later. The invariant is that every string that will be used in later steps has an available canonical representative, and this representative can replace the original because it preserves exactly the occurring $z$-substrings and their order. Hence $S'$ produces $z$ from $A'$ with $|S'|=|S|$, implying $TD_{NUC}(A,\{z\})=TD_{NUC}(A',\{z\})$. Since every canonical string has length at most $2|z|+1$, it has at most $2|z|+2$ cut positions, and the number of possible canonical strings is bounded by $2e^{4(|z|+1)^2}$. Thus, both the size of the kernel $A'$ and the number of cut positions considered at each search node are bounded by functions of $|z|$. The construction is described formally in Algorithm~\ref{alg:A'}, and the bounded-search algorithm is given in Algorithm~\ref{alg:backtrack}.

\begin{definition}
    Let $V \subseteq \mathit{SB}$. We say that $V$ is valid $ \iff \forall x, y \in V$, $x = z[i \dots j]$, $y = z[k \dots l]$ we have $j < k$ or $i > l$.
\end{definition}

Informally, $V$ is a valid set if any two of its substrings do not overlap. Let $\mathit{SV} = \{V \subseteq \mathit{SB} \vert V \text{ is valid} \}$. In Algorithm~\ref{alg:A'} we describe how to compute $A'$.

\begin{definition}\label{def:canonical}
Let $z_1,\dots,z_k$ be pairwise-disjoint substrings of $z$, and let
$f_1,\dots,f_{k+1}\in\{\epsilon,\#\}$, where $\#\notin\Sigma$ is a special
symbol. A string
$x'=f_1z_1f_2z_2\cdots f_kz_kf_{k+1}$
is called a canonical string.

Let $x\in\Sigma^*$. We say that $x'$ is a canonical form of $x$ if there
exist strings $w_1,\dots,w_{k+1}\in\Sigma^*$ such that
$x=w_1z_1w_2z_2\cdots w_kz_kw_{k+1}$ and, for every
$i\in\{1,\dots,k+1\}$, $f_i=\epsilon$ if and only if $w_i=\epsilon$, while
$f_i=\#$ if and only if $w_i\neq\epsilon$.

We refer to $z_1,\dots,z_k$ as the $z$-blocks and to
$w_1,\dots,w_{k+1}$ as the foreign blocks. In the canonical form $x'$,
each nonempty foreign block is represented by the single symbol $\#$,
whereas an empty foreign block is represented by $\epsilon$.
\end{definition}

\begin{definition}
Let $x\in\Sigma^*$ and let
$x'=f_1z_1f_2z_2\cdots f_kz_kf_{k+1}$
be a canonical form of $x$, witnessed by the decomposition
$x=w_1z_1w_2z_2\cdots w_kz_kw_{k+1}$.

A cut position $p\in\{0,\dots,|x|\}$ is said to lie inside the foreign
block $w_i$ if the cut occurs strictly between two symbols of $w_i$.
More precisely, let
$d_i=|w_1z_1w_2z_2\cdots w_{i-1}z_{i-1}|$.
Then $p$ lies inside $w_i$ if
$d_i<p<d_i+|w_i|$.
\end{definition}

\begin{definition}
Let $\Sigma$ be an alphabet and let $\#\notin\Sigma$ be a special symbol.
For any string $x\in(\Sigma\cup\{\#\})^*$, let $\mathrm{norm}(x)$ be the
string obtained by replacing every maximal run of consecutive $\#$ symbols
by a single $\#$.

For strings $x_1,\dots,x_m\in(\Sigma\cup\{\#\})^*$, we define their
normalized concatenation as
$x_1||x_2||\cdots||x_m=\mathrm{norm}(x_1x_2\cdots x_m)$.
\end{definition}

\begin{definition}
Let $S=(s_1,\dots,s_n)$ be a translocation sequence with $s_i=(x_i,y_i)\vdash(u_i,v_i)$.
We say that a string $x$ is consumed by the translocation $s_i$ if $x=x_i$ or $x=y_i$.
\end{definition}

\begin{lemma}\label{lemma:depth}
    Given an alphabet $\Sigma$ and $A, B \subset \Sigma^+$ such that $B = \{z\}$ then $TD_{NUC}(A, B) \leq |z| + 1$.
\end{lemma}
\begin{proof}
    We start by finding a string in $A$ that contains $z[1]$, the first character of $z$. Call this string $r$. We iterate over the rest of the characters of $z$ (from $2$ to $|z|$). At step $i$ we search in $A$ for a string that contains $z[i]$ and perform a translocation with it and $r$ to add $z[i]$ to $r$. After we go through all the characters in $z$, we have 3 cases: $r$ = $z$, $r$ is a prefix/suffix of $z$ or $r$ is a substring of $z$ which is not a prefix/suffix.  Thus, if we are in the first case then we have obtained the target string in exactly $|z| - 1$ steps. In the second case we have to perform one more translocation to remove a suffix/prefix from the string $r$, and in the third case we have to perform two more translocations to remove both a prefix and a suffix.
\end{proof}

\begin{algorithm}[!thb]
\caption{Compute $A'$}
\label{alg:A'}
\begin{algorithmic}[1]
    \State $A'\gets\emptyset$
    \For{$V\in\mathit{SV}$}
        \State $k\gets |V|$ and write $V=\{z_1,\dots,z_k\}$
        \For{$\pi\in P_k$} \Comment{$P_k$ is the set of permutations of length $k$}
        \For{$(f_1,\dots,f_{k+1})\in\{\epsilon,\#\}^{k+1}$}
        \State $x'\gets f_1z_{\pi(1)}f_2z_{\pi(2)}\cdots
        f_kz_{\pi(k)}f_{k+1}$
        \If{$\exists x\in A$ such that $x'$ is a canonical form of $x$}
        \State $A'\gets A'\cup\{x'\}$
        \EndIf
        \EndFor
        \EndFor
    \EndFor
    \State \Return{$A'$}
\end{algorithmic}
\end{algorithm}

Algorithm~\ref{alg:A'} constructs the kernel $A'$ by enumerating every possible canonical string determined by a valid set of pairwise-disjoint $z$-blocks, their left-to-right order, and the positions of the foreign blocks represented by $\#$ symbols. A canonical string is included in $A'$ exactly when it is a canonical form of at least one string in the original input set $A$.

\begin{definition}
    Let $S = (s_1, s_2, \dots, s_n)$ be a $B$-producing translocation sequence. For each $i, 1 \leq i \leq n$, we define $A^{i}_{min}  \subset \Sigma^+$ as the minimal multiset of strings (allowing repetitions) necessary before step $i$ to execute the remaining translocations ($s_i, s_{i+1}, \dots, s_{n}$):
    \begin{itemize}
        \item $A^{n}_{min} = \{x_{n}, y_{n}\}$
        \item $\forall i, 1 \leq i < n, A^{i}_{min}  = A^{i+1}_{min}  \setminus \{u_i, v_i\} \cup \{x_i, y_i\}$
        \item $A^{1}_{min} \subseteq A$
    \end{itemize}
    Since $A^{i+1}_{min}$ is a multiset $\forall i, 1 \leq i < n$, then the operation $A^{i}_{min}  = A^{i+1}_{min}  \setminus \{u_i, v_i\} \cup \{x_i, y_i\}$ removes a single copy of $u_i$ and $v_i$, while adding a new copy of $x_i$ and $y_i$.
\end{definition}

\begin{lemma}\label{lemma:kernel}
Let $S=(s_1,s_2,\dots,s_{|S|})$ be an optimal $B$-producing translocation
sequence. Then there exists a $B$-producing translocation sequence
$S'=(s'_1,s'_2,\dots,s'_{|S'|})$ with $|S'|=|S|$ such that, at every step
$i$, $1\leq i\leq |S'|$, the following properties hold:
\begin{enumerate}
    \item every string in ${A'}^i_{\min}$ is a canonical string;
    \label{p1}

    \item the strings in the multiset ${A'}^i_{\min}$ can be ordered as
    $r_1,\dots,r_m$ so that their normalized concatenation
    $r_1||r_2||\cdots||r_m$ is a canonical string;
    \label{p2}

    \item for every string $x\in A^i_{\min}$, there exists a string
    $x'\in {A'}^i_{\min}$ such that $x'$ is a canonical form of $x$.
    Moreover, if $x$ is consumed by $s_j$ for some
    $i\leq j\leq |S|$, then $x'$ is consumed by $s'_j$.
    \label{p3}
\end{enumerate}
\end{lemma}

\begin{proof}
We prove the existence of $S'$ by reverse induction on the step index, starting from $k = |S'| = |S|$ down to $k = 1$. For each $k\in\{1,\dots,|S'|\}$ we maintain the inductive invariant:

$P(k): $ Starting with step $k$, we can simulate all the translocations in $S$ ($s_k, s_{k+1}, \dots, s_{|S|}$) using the same number of translocations ($s'_k, s'_{k+1} \dots s'_{|S|}$) such that the properties~\ref{p1}, \ref{p2}, and~\ref{p3} are satisfied at each step $i \geq k$ and $s'_{|S|}$ produces $z$.

\noindent \textbf{Base case}. Let $|S| = n$ and $s_n = (x_n, y_n) \vdash_{px_n, py_n} (u_n, v_n)$ be the last translocation in the sequence $S$.
Since $S$ is an optimal sequence of translocations then $u_n = z$ or $v_n = z$. Without loss of generality, we assume that $v_n = z$. Thus, $z = x_n[1 \dots px_n]y_n[py_n + 1 \dots |y_n|]$. Depending on the values of $px_n, py_n$, we can define $x'_n$, $y'_n$, $u'_n$ and $v'_n$ as follows and build the translocation $s_n' = (x'_n, y'_n) \vdash_{px_n, py'_n} (u'_n, z)$:

\begin{itemize}
    \item if $px_n \geq 1$ and $py_n < |y_n|$ ($x_n[1 \dots px_n]$ is a proper prefix of $z$ and $y_n[py_n + 1 \dots |y_n|]$ is a proper suffix of $z$)
     \begin{itemize} 
        \item  if $px_n < |x_n|$, then $x_n' = x_n[1 \dots px_n]\# = z[1 \dots px_n]\#$.
            \begin{itemize}
                \item if $py_n = 0$, then  $y_n' = y_n = z[px_n + 1 \dots |z|]$, $py'_n = 0$, $u'_n = \#$ $\implies {A'}^n_{min} = \{z[1\dots px_n]\#, z[px_n + 1 \dots |z|]\}$.

                \item if $py_n > 0$, then $y_n' = \#y_n[py_n + 1\dots|y_n|] = \#z[px_n + 1 \dots |z|]$, $py'_n = 1$, $u'_n = \#$ $\implies {A'}^n_{min} = \{z[1\dots px_n]\#, \#z[px_n + 1 \dots |z|]\}$.
            \end{itemize}
         \item  if $px_n = |x_n|$, then $x_n' = x_n[1 \dots px_n] = z[1 \dots px_n]$
            \begin{itemize}
                \item if $py_n = 0$, then $y_n' = y_n = z[px_n + 1 \dots |z|]$, $py'_n = 0$, $u'_n = \epsilon$ $\implies {A'}^n_{min} = \{z[1\dots px_n], z[px_n + 1 \dots |z|]\}$.
                \item if $py_n > 0$, then $y_n' = \#y_n[py_n + 1\dots|y_n|] = \#z[px_n + 1 \dots |z|]$, $py'_n = 1$, $u'_n = \#$ $\implies {A'}^n_{min} = \{z[1\dots px_n], \#z[px_n + 1 \dots |z|]\}$.
            \end{itemize}
    \end{itemize}

    \item if $px_n = 0$, then $y_n' = \#y_n[py_n + 1 \dots |y_n|] = \#z$, $py'_n = 1$, $u'_n = \#$
    \begin{itemize}
        \item if $|x_n| > 0$, then $x'_n = \# \implies {A'}^n_{min} = \{\#, \#z\}$.
        \item if $|x_n| = 0$, then $x'_n = \epsilon \implies {A'}^n_{min} = \{\epsilon, \#z\}$.
    \end{itemize}
    \item if $py_n = |y_n|$, then $x_n' = x_n[1 \dots px_n]\# = z\#$, $u'_n = \#$
    \begin{itemize}
        \item if $|y_n| > 0$, then $y'_n = \#, py'_n = 1 \implies {A'}^n_{min} = \{z\#, \#\}$.
        \item if $|y_n| = 0$, then $y'_n = \epsilon, py'_n = 0 \implies {A'}^n_{min} = \{z\#, \epsilon\}$.
    \end{itemize}
\end{itemize}

Thus, $s'_n$ produces $z$ and $\forall px_n, py_n$ we have: $x'_n, y'_n$ are canonical strings, $x'_n || y'_n$ is a canonical string, $x'_n$ is a canonical form of $x_n$ and $y'_n$ is a canonical form of $y_n$  $\implies P(n)$ is true.

\noindent \textbf{Inductive step}. $P(k + 1) \implies P(k)$. 

By $P(k{+}1)$, all the strings in ${A'}^{k+1}_{min}$ are canonical and the strings in ${A'}^{k+1}_{\min}$ can be ordered so that their normalized concatenation is a canonical string (alternating pairwise-disjoint, nonempty $z$-blocks and separators from $\{\epsilon,\#\}$ with at most one ``\#'' between two blocks).

Let $s_k = (x_k, y_k) \vdash_{px_k, py_k} (u_k, v_k)$ be the $k$-th translocation operation from $S$. Since $S$ is optimal, at least one of $u_k$ or $v_k$ must be consumed by a subsequent translocation $s_l$ (where $l > k$). Thus, at least one of $u_k$ or $v_k$ must be in ${A}^{k+1}_{min}$ and we distinguish two main cases:
\begin{itemize}
    \item Both $u_k$ and $v_k$ are in ${A}^{k+1}_{min} \implies \exists l_1, l_2 > k$ such that $u_k \in \{x_{l_1}, y_{l_1}\}$ and $v_k \in \{x_{l_2}, y_{l_2}\} \implies \exists u'_k, v'_k \in {A'}^{k+1}_{min}$ such that $u'_k, v'_k$, $u'_k || v'_k$ are canonical strings, $u'_k$ is a canonical form of $u_k$, $u'_k \in \{x'_{l_1}, y'_{l_1}\}$,  $v'_k$ is a canonical form of $v_k$ and $v'_k \in \{x'_{l_2}, y'_{l_2}\}$. 

    \item Only one of the $u_k$, $v_k$ strings is in ${A}^{k+1}_{min}$. Without loss of generality, we assume $v_k \in {A}^{k+1}_{min}  \implies \exists l > k$ such that $v_k \in \{x_{l}, y_{l}\} \implies \exists v'_k \in {A'}^{k+1}_{min}$ such that $v'_k$  is a canonical string, $v'_k$ is a canonical form of $v_k$ and $v'_k \in \{x'_{l}, y'_{l}\}$. Additionally, we set $u'_k=\#$ if $|u_k| > 0$ or $u'_k = \epsilon$ if $|u_k| = 0 \implies u'_k$ is a canonical string, $u'_k$ is a canonical form of $u_k$ and $u'_k || v'_k$ is also a canonical string.
\end{itemize}

    Depending on the values of $px_k, py_k$, we define $x'_k$, $y'_k$, $px'_k, py'_k$ as follows and build the translocation $s_k' = (x'_k, y'_k) \vdash_{px'_k, py'_k} (u'_k, v'_k)$:
    \begin{enumerate}
        \item We obtain the values $px'_k$ and $py'_k$ by comparing the strings
        $u_k$ and $v_k$ with their canonical forms $u'_k$ and $v'_k$, respectively.
        More precisely, let ${u_k^{py_k}}^{\mathit{fb}}$ be the number of $\#$
        symbols in $u'_k$ whose corresponding foreign blocks in $u_k$ contain at
        least one character of the prefix $u_k[1\dots py_k]$, including the foreign
        block containing the cut position if the cut lies inside such a block.
        Let ${u_k^{py_k}}^{\mathit{fc}}$ be the number of characters of
        $u_k[1\dots py_k]$ belonging to these foreign blocks.
        
        Similarly, let ${v_k^{px_k}}^{\mathit{fb}}$ be the number of $\#$ symbols
        in $v'_k$ whose corresponding foreign blocks in $v_k$ contain at least one
        character of the prefix $v_k[1\dots px_k]$, and let
        ${v_k^{px_k}}^{\mathit{fc}}$ be the number of characters of
        $v_k[1\dots px_k]$ belonging to these foreign blocks.
        
        We then set
        $px'_k=px_k+{v_k^{px_k}}^{\mathit{fb}}-{v_k^{px_k}}^{\mathit{fc}}$
        and
        $py'_k=py_k+{u_k^{py_k}}^{\mathit{fb}}-{u_k^{py_k}}^{\mathit{fc}}$.
        \item We define $x'_k$ and $y'_k$ as follows:
        \begin{itemize}
            \item if $1 \leq px_k < |x_k|$ 
                \begin{itemize}
                    \item if $1 \leq py_k < |y_k|$
                        \begin{itemize}
                            \item $x'_k = v'_k[1\dots px'_k]u'_k[py'_k \dots |u'_k|]$ if $py_k$ is inside a foreign block of $u_k$, otherwise $x'_k = v'_k[1\dots px'_k]u'_k[py'_k + 1 \dots |u'_k|]$
                            \item $y'_k = u'_k[1\dots py'_k]v'_k[px'_k \dots |v'_k|]$ if $px_k$ is inside a foreign block of $v_k$, otherwise $y'_k = u'_k[1\dots py'_k]v'_k[px'_k + 1 \dots |v'_k|]$
                        \end{itemize}
                    \item if $py_k = 0$
                        \begin{itemize}
                            \item $x'_k = v'_k[1\dots px'_k]u'_k$
                            \item $y'_k = v'_k[px'_k \dots |v'_k|]$ if $px_k$ is inside a foreign block of $v_k$, otherwise $y'_k = v'_k[px'_k + 1 \dots |v'_k|]$
                        \end{itemize}
                    \item if $py_k = |y_k|$
                        \begin{itemize}
                            \item $x'_k = v'_ku'_k[py'_k \dots |u'_k|]$ if $py_k$ is inside a foreign block of $u_k$, otherwise $x'_k = v'_ku'_k[py'_k + 1 \dots |u'_k|]$
                            \item $y'_k = u'_k[1\dots py'_k]$
                        \end{itemize}
            \end{itemize}
             \item if $px_k = 0$
             \begin{itemize}
                 \item $1 \leq py_k < |y_k|$
                     \begin{itemize}
                            \item $x'_k = u'_k[py'_k \dots |u'_k|]$ if $py_k$ is inside a foreign block of $u_k$, otherwise $x'_k = u'_k[py'_k + 1 \dots |u'_k|]$
                            \item $y'_k = u'_k[1\dots py'_k]v'_k$
                        \end{itemize}
                \item $py_k = |y_k|$
                     \begin{itemize}
                            \item $x'_k = u'_k[py'_k \dots |u'_k|]$ if $py_k$ is inside a foreign block of $u_k$, otherwise $x'_k = u'_k[py'_k  + 1 \dots |u'_k|]$
                            \item $y'_k = u'_k[1\dots py'_k]$
                        \end{itemize}
                \item $py_k = 0$ - this is not a valid case in an optimal translocation sequence since it is producing exactly the same strings
                \end{itemize}
            \item if $px_k = |x_k|$ 
                \begin{itemize}
                    \item $1 \leq py_k < |y_k|$
                         \begin{itemize}
                                \item $x'_k = v'_k[1\dots px'_k]$
                                \item $y'_k = u'_k[1\dots py'_k]v'_k[px'_k \dots |v'_k|]$ if $px_k$ is inside a foreign block of $v_k$, otherwise $y'_k = u'_k[1\dots py'_k]v'_k[px'_k + 1 \dots |v'_k|]$
                            \end{itemize}
                    \item $py_k = 0$
                         \begin{itemize}
                                \item $x'_k = v'_k[1\dots px'_k]$
                                \item $y'_k = v'_k[px'_k \dots |v'_k|]$ if $px_k$ is inside a foreign block of $v_k$, otherwise $y'_k = v'_k[px'_k + 1 \dots |v'_k|]$
                            \end{itemize}
                    \item $py_k = |y_k|$ - this is not a valid case in an optimal translocation sequence since it is producing exactly the same strings
                \end{itemize}
        \end{itemize}
    \end{enumerate}
    For all possible values of $px_k$ and $py_k$, we define $x'_k$ and $y'_k$
    such that $x'_k$ is a canonical form of $x_k$ and $y'_k$ is a canonical
    form of $y_k$. Additionally, since $u'_k$, $v'_k$, and $u'_k||v'_k$ are
    canonical strings, $x'_k$, $y'_k$, and $x'_k||y'_k$ are also canonical
    strings. Moreover, the strings in
    ${A'}^{k}_{\min}=\{x'_k,y'_k\}\cup
    ({A'}^{k+1}_{\min}\setminus\{u'_k,v'_k\})$
    can be ordered so that their normalized concatenation is a canonical
    string. Therefore, $P(k)$ holds. By reverse induction from $k=n$, the lemma
    follows.
\end{proof}

\begin{algorithm}[!thb]
\caption{FPT algorithm for $TD_{NUC}(A',\{z\})$}
\label{alg:backtrack}
\begin{algorithmic}[1]
\State \textbf{Input:} kernel $A'$ (all strings canonical), target string $z$
\State \textbf{Output:} $TD_{NUC}(A',\{z\})$
\State $k \gets |z|$
\State $best \gets \infty$
\State $U_0 \gets \emptyset$ \Comment{multiset of produced strings currently available}

\Procedure{FPT}{$U,\ steps$}
  \If{$z \in A'\cup U$}
    \State $best \gets \min(best,\ steps)$
    \State \Return
  \EndIf
  \If{$steps \ge best$ \textbf{ or } $steps = k + 1$}
    \State \Return \Comment{never explore more than $|z|+1$ steps}
  \EndIf

  \ForAll{$x \in A'\cup U$}
    \ForAll{$y \in A'\cup U$}
      \For{$i \gets 0$ to $|x|$}
        \For{$j \gets 0$ to $|y|$}
          \State $(u,v) \gets$ \Call{Translocate}{$x,y,i,j$}
          \State $U_{\text{new}} \gets U$
          \If{$x \in U$}
              \State remove one copy of $x$ from $U_{\text{new}}$
          \EndIf
          \If{$y \in U_{\text{new}}$}
              \State remove one copy of $y$ from $U_{\text{new}}$
          \EndIf
          \State $U_{\text{new}} \gets U_{\text{new}} \cup \{u,v\}$
          \State \Call{FPT}{$U_{\text{new}},\ steps+1$}
        \EndFor
      \EndFor
    \EndFor
  \EndFor
\EndProcedure

\State \Call{FPT}{$U_0,\ 0$}
\State \Return $best$
\end{algorithmic}
\end{algorithm}

Algorithm~\ref{alg:backtrack} computes $TD_{NUC}(A',\{z\})$ by performing a bounded exhaustive search over contiguous translocation sequences starting from the kernel $A'$. At each search node, it considers every ordered pair of currently available strings and every pair of cut positions, applies the corresponding translocation, removes consumed copies of previously produced strings, and adds the two resulting strings to the multiset of available produced strings. By Lemma~\ref{lemma:depth}, it is sufficient to explore sequences of length at most $|z|+1$, and the minimum depth at which $z$ is obtained is the translocation distance.

\begin{theorem}\label{thm:fpt}
 Algorithms~\ref{alg:A'} and~\ref{alg:backtrack} compute
$TD_{NUC}(A,\{z\})$ in time $e^{O(q^3)}\mathrm{poly}(|I|)$, where $q = |z|$ and $|I|$ is the input size.
\end{theorem}

\begin{proof}

\textbf{Correctness}:
By Lemma~\ref{lemma:kernel}, $TD_{NUC}(A,\{z\})=TD_{NUC}(A',\{z\})$, so it is enough to prove that Algorithm~\ref{alg:backtrack} returns $TD_{NUC}(A',\{z\})$. In the algorithm, $U$ is the multiset of produced strings currently available; the initial strings in $A'$ are always available. Hence, at a state $U$, the available strings are exactly those in $A'\cup U$. The algorithm performs a bounded search over all valid contiguous translocation sequences from $A'$. Whenever $z\in A'\cup U$, it updates $best$ with the length of the corresponding sequence. By Lemma~\ref{lemma:depth}, $TD_{NUC}(A',\{z\})\leq q+1$, where $q=|z|$. Since the algorithm explores every sequence of length at most $q+1$ by trying all available ordered pairs and all cut positions, while removing consumed produced strings, it explores an optimal sequence. Therefore, it sets $best=TD_{NUC}(A',\{z\})$ and returns $TD_{NUC}(A,\{z\})$.

\noindent \textbf{Complexity}: Fix $k\in\{0,\dots,q\}$ and let $\mathit{SV}_k$ denote the family of valid sets $V\subseteq \mathit{SB}$ with $|V|=k$.
Algorithm~\ref{alg:A'} enumerates, for each $V\in \mathit{SV}_k$, each permutation $\pi\in P_k$ and each tuple $(f_1,\dots,f_{k+1})\in\{\epsilon,\#\}^{k+1}$ and adds at most one string for each such triple that has a witness in $A$. Therefore, $|A'| \le \sum_{k=0}^{q = |z|} |\mathit{SV}_k|\, k!\, 2^{k+1}$ (1). We identify a substring $z[i\dots j]$ with the half-open interval of cut positions $(i-1,\,j]$ in
$\{0,1,\dots,q\}$, where $i-1$ is the cut position immediately before the first character and $j$ is the cut
position immediately after the last character. In particular, the left endpoint is not part of the substring,
while the right endpoint is included. We bound $|\mathit{SV}_k|$ as follows. Any $V\in \mathit{SV}_k$ can be ordered from left to right along $z$, and then it is uniquely
represented by cut positions
$0\le e_1<e_2\le e_3<e_4\le \cdots \le e_{2k-1}<e_{2k}\le q$,
where the $t$-th block corresponds to the interval $(e_{2t-1},\,e_{2t}]$, i.e.\ to the substring
$z[e_{2t-1}+1 \dots e_{2t}]$. Define shifted endpoints $e'_{2t-1}=e_{2t-1}+(t-1)$ and $e'_{2t}=e_{2t}+(t-1)$ for $t=1,\dots,k$. Then,
$0\le e'_1 < e'_2 < \cdots < e'_{2k}\le q + (k - 1)$.
Hence the map $V \mapsto \{e'_1, \dots, e'_{2k}\}$ is an injection into the family of $2k$-subsets of $\{0,1,\dots,q+k-1\}$.
Consequently, $|\mathit{SV}_k| \le \binom{q+k}{2k}$ (2). Combining (1) and (2) gives $|A'| \le \sum_{k=0}^{q}\binom{q + k}{2k} k! 2^{k+1}$ (3). Using $\binom{q+k}{2k}\le \frac{(q+k)^{2k}}{(2k)!}$ and $(2k)!\ge 2^k (k!)^2$ (since
$\binom{2k}{k} = \frac{(2k)!}{(k!)^2} \ge 2^k$), we obtain for each $k$:
$\binom{q+k}{2k} k! 2^{k+1}
\le \frac{(q+k)^{2k}}{(2k)!}\, k!\, 2^{k+1}
\le (q+k)^{2k} \frac{k! 2^{k+1}}{2^k(k!)^2}
= 2 \frac{(q+k)^{2k}}{k!}$ (4). Since $k \le q$, we have $q + k \le 2(q + 1)$. Hence
$\binom{q+k}{2k} k! 2^{k+1}
\le 2 \frac{(2(q + 1))^{2k}}{k!}
= 2 \frac{(4(q+1)^2)^k}{k!}$ (5). Therefore, $o'=|A'|
\le 2 \sum_{k=0}^{q}\frac{(4(q+1)^2)^k}{k!}
\le 2 \sum_{k = 0}^{\infty}\frac{(4(q+1)^2)^k}{k!}
= 2 e^{4 (q + 1)^2}$ (6). The construction of $A'$ requires checking, for each generated canonical string, whether it has a witness in $A$. Each such check can be performed in polynomial
time in the input size $|I|$. Hence the construction of $A'$ takes
$e^{O(q^2)}\mathrm{poly}(|I|)$ time. Every canonical string over $z$ contains at most $2q + 1$ symbols, hence it has at most $2q + 2$ cut positions. At any recursion node of depth at most $q+1$, the multiset $U$ contains at most
$2(q+1)$ produced strings. Hence the number of available strings is at most
$o'+2(q+1)$, where $o'=|A'|$. For each ordered pair of available strings, the
algorithm tries at most $(2q+2)^2$ pairs of cut positions. Therefore, the
branching factor $\beta$ satisfies $\beta \leq (o'+2(q+1))^2(2q+2)^2$ (7). The recursion depth is at most $q+1$, and the work performed at each node is polynomial in $|I|$. Hence $T(q)=O(\beta^{q+1}\mathrm{poly}(|I|))$ (8). Using (6) in (7), we obtain $\beta \leq (2e^{4(q+1)^2}+2(q+1))^2(2q+2)^2 \leq 64(q+1)^2e^{8(q+1)^2}$ (9). Substituting (9) into (8) gives
$T(q)=O((64(q+1)^2e^{8(q+1)^2})^{q+1}\mathrm{poly}(|I|))
=e^{O(q^3)}\mathrm{poly}(|I|)$. Together with the construction of $A'$, the total running time is $e^{O(q^3)}\mathrm{poly}(|I|)$.
\end{proof}

\section{Open problems}

Several questions remain open. In particular, it would be interesting to study the parameterized complexity of the problem with respect only to the translocation distance. Another natural direction is to determine whether the FPT result for the non-uniform contiguous variant with $|B|=1$ extends to larger target sets, and whether analogous FPT algorithms can be obtained for the non-uniform non-contiguous variant. It would also be interesting to determine the complexity of both variants over fixed alphabets, in particular over a binary alphabet or the DNA alphabet $\{A,C,G,T\}$, since our hardness reductions use an alphabet whose size grows with the input instance. Finally, another direction is to design approximation algorithms for the non-uniform non-contiguous variant.

\bibliography{bibl}

\end{document}